\documentclass[11pt]{article}

\usepackage[margin=1in]{geometry}
\usepackage{amsmath}
\usepackage{mathrsfs}
\usepackage{amssymb}
\usepackage{graphicx}
\usepackage{extarrows}
\usepackage{algorithm}
\usepackage{algorithmicx}
\usepackage{amsthm}
\usepackage{tikz}
\usepackage[noend]{algpseudocode}
\usepackage{hyperref}
\definecolor{darkgreen}{rgb}{0,0.5,0}
\hypersetup{
    unicode=false,          
    colorlinks=true,        
    linkcolor=blue,          
    citecolor=darkgreen,        
    filecolor=magenta,      
    urlcolor=cyan           
}
\usepackage{bm}
\usepackage{authblk}

\newtheorem{theorem}{Theorem}
\newtheorem{corollary}{Corollary}

\newtheorem{definition}{Definition}
\newtheorem{conjecture}{Conjecture}
\newtheorem{claim}{Claim}
\newtheorem{lemma}{Lemma}
\newtheorem{fact}{Fact}

\newtheorem{proposition}{Proposition}
\newtheorem{example}{Example}
\newtheorem*{theorem*}{Theorem}
\newtheorem*{remark*}{Remark}

\newcommand{\eq}[1]{\hyperref[eq:#1]{(\ref*{eq:#1})}}
\renewcommand{\sec}[1]{\hyperref[sec:#1]
	{Section~\ref*{sec:#1}}}
\newcommand{\fac}[1]{\hyperref[fac:#1]
	{Fact~\ref*{fac:#1}}}	
\newcommand{\thm}[1]{\hyperref[thm:#1]
	{Theorem~\ref*{thm:#1}}}
\newcommand{\lem}[1]{\hyperref[lem:#1]{Lemma~\ref*{lem:#1}}}
\newcommand{\clm}[1]{\hyperref[clm:#1]{Claim~\ref*{clm:#1}}}

\newcommand{\prop}[1]{\hyperref[prop:#1]
	{Proposition~\ref*{prop:#1}}}
\newcommand{\prob}[1]{\hyperref[prob:#1]
	{Problem~\ref*{prob:#1}}}
\newcommand{\cor}[1]{\hyperref[cor:#1]
	{Corollary~\ref*{cor:#1}}}
\newcommand{\fig}[1]{\hyperref[fig:#1]{Figure~\ref*{fig:#1}}}
\newcommand{\tab}[1]{\hyperref[tab:#1]{Table~\ref*{tab:#1}}}
\newcommand{\alg}[1]{\hyperref[alg:#1]
	{Algorithm~\ref*{alg:#1}}}
\newcommand{\app}[1]{\hyperref[app:#1]
	{Appendix~\ref*{app:#1}}}
\newcommand{\conj}[1]{\hyperref[conj:#1]
	{Conjecture~\ref*{conj:#1}}}
\newcommand{\pro}[1]{\hyperref[pro:#1]
	{Property~\ref*{pro:#1}}}
\newcommand{\fct}[1]{\hyperref[fct:#1]
	{Fact~\ref*{fct:#1}}}

\DeclareMathOperator{\poly}{poly}

\newcommand{\MOD}{\mathsf{MOD}}

\newcommand{\PARITY}{\mathsf{PARITY}}

\renewcommand{\subparagraph}{\paragraph}

\usetikzlibrary{shapes.geometric,arrows.meta}
\usetikzlibrary{decorations.pathreplacing,decorations.markings,decorations.pathmorphing,decorations.shapes}

\tikzstyle{zig} = [decorate,decoration={snake,amplitude=.4mm,segment length=2mm}]

\usepackage{scalerel}
\DeclareRobustCommand{\stirling}{\genfrac\{\}{0pt}{}}

\title{On the Degree of Boolean Functions as Polynomials over \texorpdfstring{$\mathbb{Z}_m$}{Lg}}

\author[1]{Xiaoming Sun}
\author[1]{Yuan Sun}
\author[2]{Jiaheng Wang}
\author[2]{Kewen Wu}
\author[1]{Zhiyu Xia}
\author[1]{Yufan Zheng}
\affil[1]{Institute of Computing Technology, Chinese Academy of Sciences, China}
\affil[2]{School of Electronics Engineering and Computer Science, Peking University, China}

\date{}

\begin{document}
	
	\maketitle
	
	\begin{abstract}
	    Polynomial representations of Boolean functions over various rings such as $\mathbb{Z}$ and $\mathbb{Z}_m$ have been studied since Minsky and Papert (1969). From then on, they have been employed in a large variety of areas including communication complexity, circuit complexity, learning theory, coding theory and so on. For any integer $m\ge2$, each Boolean function has a unique multilinear polynomial representation over ring $\mathbb Z_m$. The degree of such polynomial is called \textit{modulo-$m$ degree}, denoted as $\deg_m(\cdot)$.  

        In this paper, we investigate the lower bound of modulo-$m$ degree of Boolean functions. When $m=p^k$ ($k\ge 1$) for some prime $p$, we give a tight lower bound $\deg_m(f)\geq k(p-1)$ for any non-degenerate function $f:\{0,1\}^n\to\{0,1\}$, provided that $n$ is sufficient large. When $m$ contains two different prime factors $p$ and $q$, we give a nearly optimal lower bound for any symmetric function $f:\{0,1\}^n\to\{0,1\}$ that $\deg_m(f) \geq \frac{n}{2+\frac{1}{p-1}+\frac{1}{q-1}}$. 
	\end{abstract}
	
\section{Introduction}

Given a Boolean function $f:\{0,1\}^n \to \{0,1\}$, the \emph{degree} (resp.,  \emph{modulo-$m$ degree}), denoted as $\deg(f)$ (resp., $\deg_m(f)$), is the degree of the unique\footnote{The existence and uniqueness are guaranteed by the M\"obius inversion, see e.g. \cite{gopalan2009on}.} multilinear polynomial representation of $f$ over $\mathbb{R}$ (resp., $\mathbb{Z}_m$). These complexity measures and related notions have been studied extensively since the work of Minsky and Papert~\cite{minsky1969introduction}. The polynomial representation of a Boolean function has found numerous applications in the study of query complexity (see e.g.~\cite{buhrman2002complexity}), communication complexity~\cite{buhrman2001communication,razborov2003quantum,shi2007quantum,sherstov2012strong,sherstov2008communication,montanaro2009communication,gopalan2006computing}, learning theory~\cite{kushilevitz1993learning,linial1989constant,klivans2004learning,mossel2003learning}, explicit combinatorial constructions~\cite{grolmusz2000superpolynomial,grolmusz2002constructing,gopalan2006constructing,efremenko20123}, circuit lower bounds~\cite{smolensky1987algebraic,razborov1987lower,alon2001lower,gopalan2009on} and coding theory~\cite{wilson2006lemma,yildiz2007weights,katz2005p,liu2011equivalence}, etc.

In this paper, we focus on modulo-$m$ degree of Boolean functions. 
Throughout, all Boolean functions are assumed to be \emph{non-degenerate}\footnote{A Boolean function is called non-degenerate if it depends on all its $n$ variables.}, if not specifically mentioned.
One of the complexity theoretic motivations of studying $\deg_m(f)$ is to understand the power of modular counting.
For example, the famous Razborov--Smolensky polynomial method~\cite{razborov1987lower,smolensky1987algebraic} reduces the task of proving size lower bounds for $\mathsf{AC}^0[p]$ circuits to proving a lower bound of approximate modulo-$p$ degree of the target Boolean function.
However, this approach mainly works when $p$ is a prime.\footnote{It is a folklore that $\mathsf{AC}^0[m] = \mathsf{AC}^0[\mathrm{rad}(m)]$, where $\mathrm{rad}(m)$ is the square-free part of $m$. Therefore in fact we are able to handle $\mathsf{AC}^0[q]$ circuits for any prime power $q$.}
Another example, in which $m$ can be composite, is that a $(1/2+o(1))$-inapproximability of a Boolean function $f$ by degree-$O(1)$ polynomials over $\mathbb{Z}_m$ implies that $f$ cannot be computed by $\mathsf{MAJ}_{O(1)} \circ \mathsf{MOD}_m \circ \mathsf{AND}_{O(1)}$ circuits~\cite{alon2001lower}.
In general, it has been proved important to understand the computational power of polynomials over $\mathbb{Z}_m$ for general $m$.

Towards the complexity measure $\deg_m(f)$ itself, the case when $m$ is a prime has been studied a lot in previous works. For example, one natural question is whether $\deg_m(f)$ is polynomially related to $\deg(f)$ for general $m$, as other complexity measures like decision tree complexity $\mathrm{D}(f)$ do? The answer is NO according to the parity function $\PARITY(x):=\bigoplus_{i=1}^{n}x_i$. That is, $\deg_2(\PARITY)=1$ but $\deg(\PARITY)=n$. Though this function works as a counterexample for the relationship between $\deg_2(f)$ and $\deg(f)$, it is still inspiring because its modulo-$3$ degree is large. By writing $\PARITY$ as $\frac{1}{2}-\frac{1}{2}\prod_{i=1}^{n}(1-2x_i)$ and taking modulo $3$, one can get $\deg_3(\PARITY)=n$. Actually, Gopalan et al. \cite{gopalan2009on} give the following relationship between the polynomial degrees modulo two different primes $p$ and $q$:
\[
\deg_q(f)\geq\frac{n}{\lceil\log_2 p\rceil\deg_p(f)p^{2\deg_p(f)}}.
\]
Daunting at the first glance, the inequality implies an essential fact that, as long as $\deg_p(f)=o(\log n)$, a lower bound of $\Omega(n^{1-o(1)})$ for $\deg_q(f)$ follows. Moreover, if $m$ has at least two different prime factors $p$ and $q$, then $\deg_m(f)\geq\max\left\{\deg_p(f),\deg_q(f)\right\}=\Omega(\log n)$. 

Having negated the possibility for the case of prime $m$, it is natural to study the case of composite number. The systematic study of this case was initiated by Barrington et al. \cite{BarringtonBR94}. Alas, whether $\deg_m(f)$ is polynomially related to $\deg(f)$ is still a widely open problem. Though the case $m$ being a prime power is proved to be not true in Gopalan's thesis \cite{gopalan2006computing}, we are unable to find better separation between $\deg_m(f)$ and $\deg(f)$, for $m = pq$ with $p$ and $q$ being two distinct primes, than the quadratic one given by Li and Sun~\cite{li2017modulo}. This leads to the following conjecture:

\begin{conjecture}
	\label{conj:1}
	Let $f$ be a Boolean function. If $m$ has at least two distinct prime factors, then
	\[\deg(f)=O(\poly(\deg_m(f))).\]
\end{conjecture}

Towards this conjecture, the first step is to deal with \textit{symmetric} Boolean functions. Lee et al. \cite{Lee2015Restrictions} proves that  $2\deg_{p}(f)\deg_{q}(f)>n$ for any distinct primes $p,q$ and non-trivial symmetric Boolean function $f:\{0,1\}^n\to\{0,1\}$, implying the correctness of \conj{1} in symmetric cases. Li and Sun \cite{li2017modulo} improved their bound to  $p\deg_{p}(f)+q\deg_{q}(f)>n$, which implies $\deg_{pq}(f)>\frac{n}{p+q}$. This is far from being tight; actually, as we will present later, the denominator $p+q$ can be reduced to $3.5$. 

On the tight lower bound of $\deg(f)$, Nisan and Szegedy~\cite{nisan1992degree} give the bound $\deg(f)\ge\log_2 n-O(\log\log n)$ as long as $f$ is non-degenerate.
Very recently, this bound is improved to $\deg(f)\ge\log_2 n-O(1)$ by \cite{chiarelli2018tight,wellens2019tighter}, which is tight up to the additive $O(1)$-term by the {\em address function}. 
Gathen and Roche~\cite{von1997polynomials} show that $\deg(f) \geq \deg_{p(n)}(f) \geq p(n) - 1$ for any non-trivial \emph{symmetric} Boolean function, where $p(n)$ is the largest prime below $n+2$. (Notice that the module degree gives a lower bound on the degree.) Using the currently best result on prime gaps~\cite{baker2001difference}, this gives an $n-O(n^{0.525})$ lower bound. On the other side, Gathen and Roche give a polynomial family with $\deg(f)=n-3$, and they propose \conj{2} below with a probabilistic heuristic argument:

\begin{conjecture}
	\label{conj:2}
	For any non-trivial symmetric Boolean function $f:\{0,1\}^n\to\{0,1\}$, 
	\[\deg(f)\ge n-O(1).\]
\end{conjecture}

\subparagraph*{Our Results.} In this work, we prove the following four theorems, giving better lower bounds for $\deg_m(f)$. As we have already mentioned, the gap between $\deg(f)$ and $\deg_{p^k}(f)$ can be arbitrarily large. Nevertheless, we claim that $\deg_{p^k}(f)$ cannot be too small either. This begins with symmetric functions: 

\begin{theorem}
	\label{thm:degpk_sym}
	For any prime $p$, positive integer $k$, and non-trivial symmetric function $f:\{0,1\}^n \to \{0, 1\}$, 
	\[\deg_{p^k}(f) \geq (p - 1)\cdot k\]
	when $n \geq (k-1)\varphi(p^\mu)+p^\mu-1\in O(p^2k^2)$ where $\mu=\lceil\log_p((p - 1)k - 1)\rceil$. The bound $(p-1)\cdot k$ is tight.
\end{theorem}

The proof of \thm{degpk_sym} is centered around Mahler expansion \cite{mahler1958interpolation}, which has been deemed useful in several fields of study, from analytic functions to combinatorics. Wilson \cite{wilson2006lemma} studied Mahler coefficients and related degree to period of symmetric functions. However, by introducing some more insights, we are able to give a stronger analysis to settle this case once for all. 
To be a bit more concrete, our argument (i) introduces the base-$p$ period to replace normal period, and then (ii) spans every symmetric functions into two fashions, by $\mathsf{MOD}$s or binomials, and then (iii) introduces Mahler coefficient matrix and determines its kernel. 

In addition, \thm{degpk_sym} can be extended to non-degenerate Boolean functions. We achieve this by showing that one can embed an $\omega(1)$-size non-trivial symmetric Boolean function into any non-degenerate functions by applying Erd\H os--Rado Theorem from Ramsey theory.\footnote{We note that a similar embedding argument has appeared before in \cite{alon2001lower}.} This leads to the same tight bound, provided that the input size is sufficiently large. 

\begin{theorem}
	\label{thm:degpk_all}
	For any prime $p$, positive integer $k$, and non-degenerate function $f:\{0,1\}^n \to \{0, 1\}$ with sufficiently large $n$, 
	\[\deg_{p^k}(f) \geq (p - 1)\cdot k.\]
	The bound $(p-1)\cdot k$ is tight.
\end{theorem}

Now turn to the case of non-prime-power composite $m$. The following theorem provides a lower bound on $\deg_m(f)$. 

\begin{theorem}
	\label{thm:degpq_sym}
	For any composite number $m$ with at least two different prime factors $p, q$ and any non-trivial symmetric Boolean function $f:\{0,1\}^n\to\{0,1\}$, 
	\[\deg_m(f) \geq \frac{1}{2+\frac{1}{p-1}+\frac{1}{q-1}}\cdot n.\]
\end{theorem}

Note that this bound approaches $n/2$ when $p$ and $q$ become larger and larger. It improves the $n/(p+q)$ bound in \cite{li2017modulo}. To prove this theorem, we show a stronger version of \thm{degpk_sym} for $k=1$, which requires a more elaborate analysis. Then we utilizes Periodicity Lemma~\cite{fine1965uniqueness} to obtain the desired lower bound. 

On the other hand, the next theorem shows that the bound in \thm{degpq_sym} cannot be larger than $(1+o(1))n/2$:

\begin{theorem}
	\label{thm:degpq_2}
	Let $m$ be a square-free composite number.
	There exists a symmetric Boolean function $ f:\{0,1\}^{n}\to\{0,1\}$ with arbitrarily large $n$, such that $\deg_{m}(f)\leq n/2 + o_{m}(n)$.\footnote{The subscript ``$m$'' in the $o(\cdot)$ notation means that the hidden factor depends on $m$.}
\end{theorem}

\subparagraph*{Organization.} In \sec{pre}, we give necessary definitions and concepts. Then we give the proofs of \thm{degpk_sym} and \thm{degpk_all} respectively in \sec{pksym} and \sec{pknd}. In \sec{degpq_bound} we prove \thm{degpq_sym}, and in \sec{degpq_tight} we prove \thm{degpq_2}. Finally, we conclude the paper in \sec{con}.

\section{Preliminaries}\label{sec:pre}

We denote $\{1,2,\dots,n\}$ as $[n]$ throughout this paper. $\varphi(\cdot)$ denotes Euler's totient function. Notation $\log^{\circ k}(n)$ is defined as $\underbrace{\log\log\cdots\log}_{k} n$, and $\log^*(n)$ is for the iterated logarithm, that is,  $\min\{k:\log^{\circ k}(n)\leq 1\}$. 

\subsection{Basics of Boolean Functions}

An $n$-bit \textit{Boolean function} $f(x)$ is a mapping from $\{0,1\}^n$ to $\{0,1\}$. Sometimes we write $\bm{x}$ to indicate the $n$-dimensional 0-1 vector corresponding to string $x \in \{0,1\}^n$. The following operation will be frequently used: Suppose $x\in\{0,1\}^n$ is a (input) string, and $S\subseteq[n]$ is a set of indices. Denote the string obtained by flipping all bits in $x$ whose indices are in $S$ as $x^{\oplus S}$. As a common practice, $x^{\oplus\{i\}}$ is abbreviated as $x^{\oplus i}$. 

Here list some subclasses of Boolean functions, which we will frequently deal with later:
\begin{itemize}
	\item A Boolean function is called \textit{non-trivial} if it is not a constant.
	\item A Boolean function is called \textit{non-degenerate} if its value depends on all input bits. In other words, there does not exist such $t$ that, for every $x\in\{0,1\}^n$ the equality $f(x)=f(x^{\oplus t})$ holds. Such bit, if exists, is also known as \textit{dumb} bit. 
	\item A Boolean function is called \textit{symmetric}, if $f(x) = f(y)$ for any $x,y$ satisfying $|x|=|y|$. Here $|x|$ denotes the Hamming weight of $x$, i.e., number of $1$'s. 
\end{itemize}

There exists a unique polynomial representing $f$ over $\mathbb{Z}_m$ or $\mathbb{Z}$. More formally:
\begin{fact}
    For any Boolean function $f:\{0,1\}^n\!\to\!\{0,1\}$, the unique polynomial
    \[
        \sum_{a\in\{0,1\}^n}f(a)\prod_{i=1}^n((2a_i-1)x_i+1-a_i)=:\sum_{S\subseteq [n]}c_S \prod_{i\in S}x_i
    \]
    represents $f$ over $\mathbb{Z}$. On top of this, the polynomial $\sum_{S\subseteq [n]}(c_S\bmod{m}) \prod_{i\in S}x_i$ represents $f$ over $\mathbb Z_m$. 
\end{fact}

\begin{definition}
	The degree (resp., modulo-$m$ degree) of a Boolean function $f$, denoted by $\deg(f)$ (resp., $\deg_m(f)$), is the degree of the polynomial representing $f$ over $\mathbb{Z}$ (resp., $\mathbb{Z}_m$). 
\end{definition}

This measure has some simple but useful properties. The following fact is a consequence of the Chinese Remainder Theorem; see \cite[{Fact 5}]{li2017modulo}. 

\begin{fact}\label{fac:divide_dominate}
	Suppose $f:\{0,1\}^n\to\{0,1\}$ is a Boolean function, and $m,m'$ are coprime. Then $\deg_{m'm}(f)=\max\{\deg_m(f),\deg_{m'}(f)\}$. 
\end{fact}

With some input bits fixed, the degree of a Boolean function may decrease. This can be easily derived by substituting those variables with their values. More formally, we define the \textit{restriction} of Boolean functions and restate this fact below. 

\begin{definition}[Restriction]
	Suppose $f:\{0,1\}^n\to\{0,1\}$ is a Boolean function, $S\subseteq[n]$ is a set of indices, and there is a mapping $\sigma:[n]\backslash S\to\{0,1\}$. For every $i\notin S$, fix the $i$-th bit in the input of $f$ to be $\sigma(i)$ to obtain a new Boolean function with input size $|S|$. We call it the restriction of $f$ over $\sigma$, denoted as $f|_\sigma$.
\end{definition}

\begin{fact}\label{fac:restrict}
	Suppose $f:\{0,1\}^n\to\{0,1\}$ is a Boolean function. For any integer $m\geq 2$ and restriction $f|_\sigma$, we have $\deg_m(f)\geq\deg_m\left(f|_\sigma\right)$.
\end{fact}

A common complexity measure, the \textit{sensitivity}, will be used in \sec{pknd}. Simon gave a lower bound on this measure~\cite{simon1983tight}. 

\begin{definition}[Sensitivity]
	Given a Boolean function $f:\{0,1\}^n\to\{0,1\}$ and an input $x$, we say bit $i$ is sensitive if $f(x)\neq f(x^{\oplus i})$. The sensitivity of $f$ on input $x$ is $s(f,x):=|\{i: i\in [n], f(x)\neq f(x^{\oplus i})\}|$. The sensitivity of $f$ is then defined as $s(f):=\max_{x}s(f,x)$. 
\end{definition}

\begin{theorem}[\cite{simon1983tight}]\label{thm:simon}
	For any non-degenerate Boolean function $f:\{0,1\}^n\to\{0,1\}$, we have
	\[
	s(f)\geq\frac{1}{2}\log n-\frac{1}{2}\log\log n+\frac{1}{2}.
	\]
\end{theorem}

\subsection{Periodicity and Mahler Expansion}\label{sec:period_mahler}

We consider symmetric Boolean functions in this section.
For a symmetric function $f:\{0,1\}^n \to \{0,1\}$, clearly there exists a unique $F:\{0,\dots,n\} \to \{0,1\}$, called the \emph{univariate version of $f$}, such that $f(x) = F(|x|)$ for every $x$.
We call $f$ (and $F$) \textit{$\ell$-periodic}, if $F(a)=F(b)$ for any $0 \leq a,b \leq n$ satisfying $\ell\,|\,a-b$. For example, $f$ is trivially $\ell$-periodic for any $\ell>n$. We are also interested in integer power period length.
Hence we introduce the following definition. 

\begin{definition}[Base-$m$ period]
	Assume $f:\{0,1\}^n\to\{0,1\}$ is a symmetric Boolean function. The base-$m$ period is the minimum $\ell$ such that $\ell$ is a power of $m$, and $f$ is $\ell$-periodic. Denote it as $\pi_m(f)$. 
\end{definition}

Here are some concrete examples for a clearer illustration. 

\begin{itemize}
	\item The not-all-equal $\mathsf{NAE}$ function is defined as $\mathsf{NAE}_n(x_1,\dots,x_{n}):=\mathbb{I}[\exists i,j\text{ s.t. }x_i\neq x_j]$. Then $\pi_3\left(\mathsf{NAE}_3\right)=3$ while $\pi_3\left(\mathsf{NAE}_4\right)=9$. That is, $\pi_m(f)$ may be larger than $n$.
	\item If $f$ is a trivial function, then $\pi_3(f)=1$. 
\end{itemize}

One may write $F$ as a univariate polynomial over $\mathbb{Q}$, but it will not always induce a polynomial when we move to work on $\mathbb{Z}_m$, like what $f$ does. Fortunately, the following representation, also known as \textit{Mahler expansion} \cite{wilson2006lemma}, serves the purpose similar to polynomial representation.

\begin{theorem}[Mahler expansion]
    Assume that $f:\{0,1\}^n\to\{0,1\}$ is a symmetric Boolean function, and $F$ is the corresponding univariate version. Let $d:=\max\{n,m-1\}$. Then there exists a unique sequence $\alpha_0, \alpha_1, \cdots, \alpha_d\in\mathbb{Z}_m$ such that
    \[
    \sum_{j=0}^{d}\alpha_j\binom{t}{j}=\left\{
    \begin{array}{cl}
        F(t), & 0\leq t\leq n; \\
        0, & n<m-1\mathrm{~and~}n<t<m.
    \end{array}
    \right.
    \]
    We call $\sum_{j=0}^{d}\alpha_j\binom{t}{j}$ the Mahler expansion of $F$ over $\mathbb{Z}_m$, and $\alpha_j$ the $j$-th Mahler coefficient.
\end{theorem}

There are some connections between polynomial degree and Mahler expansion. Over ring $\mathbb{Z}_m$, let $d^*:=\max\{\ell:\alpha_\ell\not \equiv 0 \pmod{m}, \ell\leq n\}$. If we only take $0$-th to $d^*$-th terms in the Mahler expansion to get $\hat{F}(t)=\sum_{j=0}^{d^*}\alpha_j\binom{t}{j}$, then $\hat{F}(|x|)=F(|x|)$ for all $x\in\{0,1\}^n$, which implies

\begin{fact}\label{fct:degm_mahler}
	$\deg_m(f)=\max\{\ell:\alpha_\ell\not \equiv 0 \pmod{m}, \ell\leq n\}.$
\end{fact}

\begin{remark*}
	The fact above does not hold if we take away the condition $\ell\leq n$. The next example shows that on $\mathbb{Z}_m$, the existence of high-order non-zero Mahler coefficient does not imply high degree, if the input length is too short. 
	
	\begin{example}
		Let $n=2$ and $f(x)=-x_0x_1+x_0+x_1=x_0\lor x_1$. On $\mathbb{Z}_5$, one can verify its Mahler expansion is
		\[
		f(x)=\binom{|x|}{1}+4\binom{|x|}{2}+2\binom{|x|}{4}.
		\]
		But $\deg_5(f)=2$. 
	\end{example}
	
	This phenomenon does not come from nowhere; intuitively speaking, in the Mahler expansion over $\mathbb{Z}_m$, one may need to imitate Lagrange-style interpolation for $|x|>n$ , and hence introduce some high-order terms. (Although $|x|$ can never be above $n$, we need to utilize $\mathsf{MOD}$ functions later in this paper, which requires $F(t)$ to be zero for $n<t<m$.)
	
\end{remark*}

Wilson~\cite{wilson2006lemma} showed the following result about the degree and Mahler expansion of symmetric Boolean functions, given the base-$p$ period.

\begin{theorem}[{\cite[{Lemma 1}]{wilson2006lemma}}]
	\label{thm:wilson1}
	Let $p$ be a prime, and $t,k$ be positive integers. Assume $f:\{0,1\}^n\to\{0,1\}$ is a symmetric Boolean function, and $\{\alpha_\ell\}$ are its Mahler coefficients over $\mathbb{Z}_{p^k}$. If $f$ is $p^t$-periodic, then
	\[\deg_{p^k}(f)\leq(k-1)\cdot\varphi(p^t)+p^t-1.\]
	In addition, for any positive integer $j$ and $\ell\geq j\cdot\varphi(p^t)+p^t$, we have $\alpha_\ell\equiv 0\pmod{p^j}$.
\end{theorem}

\subsection[MOD and Its Mahler Expansion over Zpk]{$\mathsf{MOD}$ and Its Mahler Expansion over $\mathbb{Z}_{p^k}$}\label{sec:moddeg}

We first look into the Mahler expansion of weight modular functions. This special case is illuminating in our later proofs. The $\mathsf{MOD}$ function is defined as
\[
\mathsf{MOD}_n^{c,m}(x):=\mathbb{I}\big[|x|\equiv c\pmod{m} \big]\in\{0,1\},
\]
where $n\geq m-1$ denotes the length of input $x$, and $\mathbb{I[\cdot]}$ is the indicator function. The following theorem gives the degree of $\mathsf{MOD}_n^{0,p^t}$. 

\begin{theorem}[{\cite[Theorem 10]{wilson2006lemma}}]
	Let $p$ be a prime, and $t,k$ be positive integers. Denote $d:=(k-1)\cdot\varphi(p^t)+p^t-1$. Then for any $n\geq d$, we have
	\[\deg_{p^k}(\mathsf{MOD}_n^{0,p^t})=d.\]
\end{theorem}

In fact, we can achieve a more general result by further analysis. 
Fix $n$, $p$, $t$ and $k$.
Notate the Mahler coefficient of $\mathsf{MOD}_n^{a,p^t}$ over $\mathbb{Z}_{p^k}$ as $\alpha_\ell^{(a,p^t)}$ i.e., $\mathsf{MOD}_n^{a,p^t}(x)=\sum_{j=0}^{d}\alpha_j^{(a,p^t)}\binom{|x|}{j}$. Moreover, $\mathsf{MOD}_n^{a,p^t}$ can also be represented with $\alpha_\ell^{(0,p^t)}$ as
\[\mathsf{MOD}_n^{a,p^t}(x)=\sum_{j=0}^{d}\alpha_j^{(0,p^t)}\binom{|x|-a}{j}.\]
Then expand each $\binom{|x|-a}{j}$ by Vandermonde convolution to get
\begin{equation}\label{equ:coef_vande}
\alpha_\ell^{(a,p^t)}=\sum_{i=0}^{d-\ell}\binom{-a}{i}\alpha_{i+\ell}^{(0,p^t)}.
\end{equation}
Specially, by setting $\ell=d$, we get $\alpha_d^{(a,p^t)}=\alpha_d^{(0,p^t)}$. This equation generalizes the theorem to all remainders. 

\begin{corollary}\label{cor:mod_deg}
	Let $p$ be a prime, and $t$ and $k$ be positive integers. Denote $d:=(k-1)\cdot\varphi(p^t)+p^t-1$. For any $n\geq d$ and $0\leq a<p^t$, we have
	\[\deg_{p^k}(\mathsf{MOD}_n^{a,p^t})=d.\]
\end{corollary}

\section[Lower Bound of degpk(f)]{Lower Bound of $\deg_{p^k}(f)$}\label{sec:degpk}

By identifying the degree of $\MOD_n^{i, p^t}$ over $\mathbb Z_{p^k}$, we show that the degree of all $p^t$-periodic functions is constantly small since they can be spanned by $\{\MOD_n^{j,p^t}\}_{j = 0}^{p^t - 1}$. In \sec{pksym}, we prove that the degree of any $p^t$-periodic (but not $p^{t - 1}$-periodic) function will not decrease too much from $(k - 1)\cdot \varphi(p^t) + p^t - 1$ during the spanning, despite the cancellation of the high-order coefficients. By a Ramsey-type argument in \sec{pknd}, we further extend our lower bound to all non-degenerate Boolean function with sufficiently many input bits.

\subsection[Proof of Theorem 1]{Symmetric Functions -- Proof of \thm{degpk_sym}}\label{sec:pksym}

We begin with the periodicity of symmetric Boolean functions with low degree. In our proof, the following Lucas's Theorem is important. 

\begin{theorem}[Lucas]
	Let $n,m\in\mathbb{N}$, and $p$ be a prime. Assume in base $p$, $n$ and $m$ can be represented as $n=(n_v n_{v-1} \cdots n_0)_p$ and $m=(m_v m_{v-1} \cdots m_0)_p$ (the number with fewer digits are padded with $0$). Then
	\[\binom{n}{m}\equiv\binom{n_v}{m_v}\binom{n_{v-1}}{m_{v-1}}\cdots \binom{n_0}{m_0} \pmod{p}.\]
\end{theorem}

The next lemma indicates the periodicity of symmetric Boolean functions with low degree. 

\begin{lemma}\label{lem:lowdeg}
	Let $f:\{0,1\}^n\to\{0,1\}$ be a symmetric Boolean function. For prime $p$ and positive integers $t$ and $k$, if $\deg_{p^k}(f)\leq p^t-1$, then $f$ is $p^t$-periodic.
\end{lemma}

\begin{proof}
	Denote $d:=\deg_{p^k}(f)$, and suppose $\alpha_\ell$ are the Mahler coefficients of $f$ over ring $\mathbb{Z}_{p^k}$, i.e., $f(x)=\sum_{j=0}^d\alpha_j\binom{|x|}{j}\bmod p^k$. According to Lucas's Theorem, if $a\equiv b\pmod{p^t}$, then for any $0\leq j\leq p^t-1$, we have $\binom{a}{j}\equiv\binom{a_v}{j_v}\cdots\binom{a_t}{j_t}\binom{a_{t-1}}{j_{t-1}}\cdots\binom{a_0}{j_0}\pmod{p}$. Here, $a_i$ (resp. $b_i$ and $j_i$) is the representation of $a$ (resp. $b$ and $j$) in the base $p$. Note that $j_i=0$ for any $i\geq t$. Hence $\binom{a}{j}\equiv\binom{a_{t-1}}{j_{t-1}}\cdots\binom{a_0}{j_0}\pmod{p}$. For the same reason,  $\binom{b}{j}\equiv\binom{b_{t-1}}{j_{t-1}}\cdots\binom{b_0}{j_0}\pmod{p}$. In addition, $a$ and $b$'s last $t$ digits are the same as $a\equiv b\pmod{p^t}$. Thus $\sum_{j=0}^d\alpha_j\binom{a}{j}\bmod p = \sum_{j=0}^d\alpha_j\binom{b}{j}\bmod p$. By the definition that $f(x)\in\{0,1\}$, we get $\sum_{j=0}^d\alpha_j\binom{a}{j}\bmod p^k = \sum_{j=0}^d\alpha_j\binom{b}{j}\bmod p^k$.
\end{proof}

	Next, provided $\pi_p(f)$, we give some lower bounds on the degree by the following two lemmas, conditioned on that $n$ is large enough. Together with the lemma above we lead to a contradiction, and our theorem follows eventually. Before continuing, notice again that the value $f(x)$ is related only to the Hamming weight of $x$, and thus $f(x)=\sum_{0\leq i\leq n:F(i)=1}\mathsf{MOD}_{n}^{i,n+1}(x)$. Specially, if $f(x)$ is $t$-periodic, then we can write $f(x)$ as $f(x)=\sum_{0\leq i\leq t-1:F(i)=1}\mathsf{MOD}_{n}^{i,t}(x)$ and apply \cor{mod_deg}. 

\begin{lemma}\label{lem:pper1}
	Assume $p$ is a prime, and $f:\{0,1\}^n\to\{0,1\}$ is a symmetric Boolean function of period $p$. If for some positive integer $k$, it holds that $n\geq(p-1)\cdot k$, then
	\[\deg_{p^k}(f)=(p-1)\cdot k.\]
\end{lemma}

\begin{proof}
	Expand $f(x)$ to get $\sum_{0\leq i\leq p-1:F(i)=1}\mathsf{MOD}_{n}^{i,p}(x)$. According to \cor{mod_deg}, the degree of each term in the summation is $\deg_{p^k}\left(\mathsf{MOD}_{n}^{i,p}\right)=(p-1)\cdot k=:d$. On the other hand, \sec{moddeg} says the coefficient of degree $d$ term is identical in the polynomial of every $\mathsf{MOD}$. As $f(x)$ is non-trivial, the number of terms in the summation, denoted as $N$, is neither $0$ nor $p$. Therefore, the degree $d$ term in $f(x)$ has coefficient $N\cdot\alpha_d^{(0,p)}\not\equiv 0\pmod{p}$, implying $\deg_{p^k}(f)=d$ as desired.
\end{proof}

\begin{lemma}\label{lem:pper2}
	Assume $p$ is a prime, and $f:\{0,1\}^n\to\{0,1\}$ is a symmetric Boolean function with $\pi_p(f)=p^t$. Here $t\geq 1$ and $n\geq (k-1)\cdot\varphi(p^t)+p^t-1$. Then
	\[\deg_{p^k}(f) \geq (k - 2)\cdot \varphi(p^t) + p^t.\]
\end{lemma}

\begin{proof}

	Consider over the ring $\mathbb{Z}_{p^k}$. Let $d:=(k-1)\cdot\varphi(p^t)+p^t-1$. Provided $p^t$, we abbreviate $\alpha_{\ell}^{(j,p^t)}$, the $\ell$-th Mahler coefficients of $\mathsf{MOD}_{n}^{j,p^t}$, as $\alpha_{\ell}^{(j)}$ for convenience. According to \cor{mod_deg}, we have $\deg_{p^k}\left(\mathsf{MOD}_{n}^{j,p^t}\right)=d$. 
	
	For all $0\leq i\leq p^t-1$, \thm{wilson1} implies the fact that $\alpha_\ell^{(i)}$ can be divided by $p^{k-2}$ when $\ell\geq (k-2)\cdot\varphi(p^t)+p^t=d-\varphi(p^t)+1$. Therefore, we divide every such coefficient by $p^{k-2}$, and then take the remainder modulo $p$ to get $\tilde{\alpha}_{\ell}^{(i)}:=\left(\alpha_\ell^{(i)}/p^{k-2}\right)\bmod p$, where $d-\varphi(p^t)+1\leq \ell\leq d$. All these $\tilde{\alpha}_{\ell}^{(i)}$ forms the matrix
	\[{\bm S}:=
	\begin{bmatrix}
	\tilde\alpha_d^{(0)} & \cdots & \tilde\alpha_d^{(p^t-1)} \\
	\vdots & \ddots & \vdots \\
	\tilde\alpha_{d - \varphi(p^t) + 1}^{(0)} & \cdots & \tilde\alpha_{d - \varphi(p^t) + 1}^{(p^t-1)}
	\end{bmatrix}
	\in\mathbb{F}_p^{\varphi(p^t)\cdot p^t}.\]
	Take its first $\varphi(p^t)$ columns to get a square matrix
	\[{\bm S'}:=
	\begin{bmatrix}
	\tilde\alpha_d^{(0)} & \cdots & \tilde\alpha_d^{\left(\varphi(p^t)-1\right)} \\
	\vdots & \ddots & \vdots \\
	\tilde\alpha_{d - \varphi(p^t) + 1}^{(0)} & \cdots & \tilde\alpha_{d - \varphi(p^t) + 1}^{\left(\varphi(p^t)-1\right)}
	\end{bmatrix}
	.\]
	When $d-\varphi(p^t)+1\leq \ell\leq d$, divide both sides of Equation (\ref{equ:coef_vande}) by $p^{k-2}$ and take the remainder modulo $p$ to get
	\[\tilde\alpha_\ell^{(a)}=\sum_{j=0}^{d-\ell}\binom{-a}{j}\tilde\alpha_{j+\ell}^{(0)},\]
	which leads to the following decomposition of matrix ${\bm S'}$:
	\[{\bm S'}
	=
	\begin{bmatrix}
	\tilde\alpha_d^{(0)} & & \\
	\vdots & \ddots &  \\
	\tilde\alpha_{d-\varphi(p^t)+1}^{(0)}& \cdots & \tilde\alpha_d^{(0)} 
	\end{bmatrix}
	\cdot
	\begin{bmatrix}
	\binom{0}{0} & \cdots & \binom{1-\varphi(p^t)}{0} \\
	\vdots & \ddots & \vdots \\
	\binom{0}{\varphi(p^t)-1} & \cdots & \binom{1-\varphi(p^t)}{\varphi(p^t)-1} 
	\end{bmatrix}
	=:{\bm T}\cdot{\bm C}.
	\]
	As $\tilde\alpha_d^{(0)}\neq 0$, the first matrix has determinant $\det({\bm T})\neq 0$. The latter one, consisting of binomial coefficients, is also invertible. We will prove this fact later. Eventually, $\mathrm{rank}({\bm S'})=\varphi(p^t)$, so the kernel of $\bm S$ has dimension $\dim\ker{\bm S}=(p^t)-\varphi(p^t)=p^{t-1}$. 
	
	On one hand, because $f(x)$ is $p^t$-periodic, we can expand it by $\{\mathsf{MOD}_n^{j,p^t}\}_{j=0}^{p^t-1}$ with coefficients $w_j$. 
	\begin{equation}\label{equ:fexpand}
	f(x)=\sum_{j=0}^{p^t-1}w_j\mathsf{MOD}_{n}^{j,p^t}(x)=\sum_{j=0}^{p^t-1}\left(w_j\sum_{\ell=0}^{d}\alpha_{\ell}^{(j)}\binom{|x|}{\ell}\right)
	=\sum_{\ell=0}^{d}\left(\left(\sum_{j=0}^{p^t-1}w_j\alpha_{\ell}^{(j)}\right)\binom{|x|}{\ell}\right).
	\end{equation}
	
	On the other hand, $\mathsf{MOD}_n^{i,p^{t-1}}$ can also be spanned by $\{\mathsf{MOD}_n^{j,p^t}\}$. In other words, assume $w_j^{(i)}:=\mathbb{I}[i\equiv j\pmod{p^{t-1}}]$, then
	\begin{equation}\label{equ:pt1_lowdeg}
	\mathsf{MOD}_n^{i,p^{t-1}}
	=\sum_{j=0}^{p^t-1}w_j^{(i)}\mathsf{MOD}_n^{j,p^t}
	=\sum_{\ell=0}^{d}\left(\left(\sum_{j=0}^{p^t-1}w_j^{(i)}\alpha_{\ell}^{(j)}\right)\binom{|x|}{\ell}\right).
	\end{equation}
	We claim $\deg_{p^k}(\mathsf{MOD}^{i,p^{t-1}}_n)\leq d-\varphi(p^t)$, and because $n\geq d$, the coefficients of highest $\varphi(p^t)$ terms in its Mahler expansion (i.e., from degree $d-\varphi(p^t)+1$ to degree $d$) are all zero. This is due to \cor{mod_deg}, where we have
	\begin{align*}
	\deg_{p^k}(\mathsf{MOD}^{i,p^{t-1}}_n)&=(k-1)\varphi(p^{t-1})+p^{t-1}-1\\
	&=p^{t-2}\left((k-1)(p-1)+p\right)-1\\
	&\leq p^{t-2}\left((k-1)(p^2-p)+p\right)-1\\
	&=d-\varphi(p^t).
	\end{align*}
	Further, if we set column vector ${\bm w}^{(i)}=\left(w_0^{(i)},...,w_{p^t-1}^{(i)}\right)^{\top}$ where $0\leq i\leq p^{t-1}-1$, then the above fact, together with Equation (\ref{equ:pt1_lowdeg}), indicates the following equation:
	\[{\bm S}{\bm w}^{(i)}=
	\left[\sum\limits_{j=0}^{p^t-1}w_j^{(i)}\tilde\alpha_{d}^{(j)},
	\sum\limits_{j=0}^{p^t-1}w_j^{(i)}\tilde\alpha_{d-1}^{(j)},
	\cdots,
	\sum\limits_{j=0}^{p^t-1}w_j^{(i)}\tilde\alpha_{d-\varphi(p^t)+1}^{(j)}
	\right]^{\top}
	={\bm 0}.
	\]
	One can verify that $\left\{{\bm w}^{(0)},\cdots,{\bm w}^{(p^{t-1}-1)}\right\}$ is linear independent, and therefore they form a base of the kernel of ${\bm S}$ i.e., $\ker{\bm S}=\mathrm{span}\left\{{\bm w}^{(0)},\cdots,{\bm w}^{(p^{t-1}-1)}\right\}$. 
	
	However, $f(x)$ is not $p^{t-1}$-periodic because $\pi_p(f)=p^t$. This means $f(x)$ cannot be written as the sum of some $\mathsf{MOD}_{n}^{i,p^{t-1}}$. Thus ${\bm w}:=\left(w_0,...,w_{p^t-1}\right)^{\top}\notin\mathrm{span}\left\{{\bm w}^{(0)},\cdots,{\bm w}^{(p^{t-1}-1)}\right\}$ i.e., ${\bm S}{\bm w}\neq{\bm 0}$. This leads to the existence of $D\in\left[d-\varphi(p^t)+1,d\right]$ such that $\sum_{j=0}^{p^t-1}w_j^{(i)}\tilde\alpha_{D}^{(j)}\neq 0$. Namely, the degree $D$ term in Equation (\ref{equ:fexpand}) exists as desired. 
\end{proof}

To finish the proof of this lemma, we show that the matrix ${\bm C}$ is invertible, due to the following proposition. 

\begin{proposition}
	$\det({\bm C})=\pm 1.$
\end{proposition}

\begin{proof}
	In fact, 
	\[{\bm C}=\mathrm{diag}\{1,-1,1,-1,...(-1)^{m-1}\}\cdot
	\begin{bmatrix}
	1&1&\binom{1}{1}&\cdots&\binom{m-2}{m-2}\\
	0&\binom{1}{0}&\binom{2}{1}&\cdots&\binom{m-1}{m-2}\\
	0&\binom{2}{0}&\binom{3}{1}&\cdots&\binom{m}{m-2}\\
	\vdots&\vdots&\vdots&\ddots&\vdots\\
	0&\binom{m-1}{0}&\binom{m}{1}&\cdots&\binom{2m-3}{m-2}\\
	\end{bmatrix}
	\]
	where $m:=\varphi(p^t)$. Denote the second matrix as ${\bm C}'$. Take Row $m$ of ${\bm C}'$ and subtract Row $m-1$ from it. Then take Row $m-1$ and subtract Row $m-2$ from it. $\cdots$ Take Row $3$ and subtract Row $2$ from it. Then $C'$ has been transformed to
	\[\begin{bmatrix}
	1&1&\binom{1}{1}&\cdots&\binom{m-2}{m-2}\\
	0&\binom{1}{0}&\binom{2}{1}&\cdots&\binom{m-1}{m-2}\\
	0&0&\binom{2}{0}&\cdots&\binom{m-1}{m-3}\\
	\vdots&\vdots&\vdots&\ddots&\vdots\\
	0&0&\binom{m-1}{0}&\cdots&\binom{2m-4}{m-3}\\
	\end{bmatrix}.\]
	Keep repeating this step, and ${\bm C}'$ will be transformed into an upper triangular matrix, whose diagonal elements are all $1$. 
\end{proof}

	Now we are ready to prove \thm{degpk_sym}. 

\begin{proof}[Proof of \thm{degpk_sym}]
	Note that \lem{pper1} provides tight instances, so below we are going to prove the inequality. 
	
	Assume towards contradiction that there exists $f:\{0,1\}^n\to\{0,1\}$ satisfying $\deg_{p^k}(f)<(p-1)\cdot k$, and $n \geq (k-1) \cdot \varphi(p^{\mu}) + p^{\mu}-1$ where $\mu = \lceil\log_p((p - 1)k - 1)\rceil$. \lem{lowdeg} tells us that $f$ is $p^{\mu}$-periodic. The non-triviality of $f$ indicates that there exists $1\leq t\leq\mu$ such that $\pi_p(f)=p^t$. 
	
	If $t=1$, then according to \lem{pper1} we have $\deg_{p^k}(f)=(p-1)\cdot k$. If $t\geq 2$, then \lem{pper2} indicates
	\begin{align*}
	\deg_{p^k}(f)&\geq (k-2)\cdot\varphi(p^t)+p^t\\
	&\geq (k-2)\cdot\varphi(p^2)+p^2\\
	&=k\cdot (p-1)^2+2p-p^2+(p-1)\cdot k\\
	&\geq (p-1)^2+2p-p^2+(p-1)\cdot k\\
	&> (p-1)\cdot k.
	\end{align*}
	Both cases lead to contradiction.
\end{proof}

\subsection[Proof of Theorem 2]{Non-degenerate Functions -- Proof of \thm{degpk_all}}\label{sec:pknd}

For the general non-degenerate case, our key idea is to embed a symmetric Boolean function into it, and then apply \thm{degpk_sym}. The following lemma is crucial. 

\begin{lemma}\label{lem:embed}
	There exists a monotone increasing function $r(n)=\omega(1)$ such that the following holds. 
	Let $f:\{0,1\}^n\to\{0,1\}$ be a non-degenerate Boolean function. 
	Then there exists a set of indices $S\subseteq[n]$ with $|S|\geq r(n)$, and a mapping $\sigma:[n]\backslash S\to\{0,1\}$ such that $f|_\sigma$ is a non-trivial symmetric Boolean function. 
\end{lemma}

Generally speaking, with this lemma in hand, every $h(n)$ lower bound on complexity measure of symmetric functions that is monotone decreasing w.r.t. restriction (e.g., \fac{restrict}) can yield an $h(r(n))$ lower bound on that of all non-degenerate functions. In the setting of modulo-$p^k$ degree, the bound $h(n)$ is a constant function (when $n$ is large than some threshold), so we can get the same bound, except that the threshold for $n$ blows up. However, as indicated by our proof, the function $r(n)$ grows extraordinary slow (approximately the square root of iterated logarithm of $n$). 

First, let us see how to utilize this lemma in proving \thm{degpk_all}. 

\begin{proof}[Proof of \thm{degpk_all}]
    As long as \lem{embed} holds, by \fac{restrict} and \thm{degpk_sym}, if $n\geq r^{-1}((k - 1)\cdot \varphi(p^\mu) + p^\mu - 1)$, then $\deg_{p^k}(f)\ge(p-1)\cdot k$, deriving the desired lower bound. In addition, any symmetric function with period $p$ (and input long enough) still serves as an instance with $\deg_{p^k}(f)=(p-1)\cdot k$. 
\end{proof}

In the rest part of this section we prove \lem{embed}. 

For convenience, we introduce the following notation. If $x_i=1$ for every $i\in S\subseteq[n]$, then define $\mathrm{DOWN}(S, x, k) := \left\{x^{\oplus T} \mid T \subseteq S,|T|=k\right\}$. Intuitively speaking, it is the set of strings obtained by flipping $k$ bits, whose indices are in $S$, of $x$ from $1$ to $0$. 

According to \thm{simon}, there exists $\tilde{x}$ such that $s(f,\tilde{x})=\Omega(\log n)$. Without loss of generality we assume the set $\{i\in[n]:\tilde{x}_i=1,f(\tilde{x})\neq f(\tilde{x}^{\oplus i})\}$, defined as $S_0$, is of cardinality $\Omega(\log n)$. Recursively define $S_t$ to be the largest set satisfying the following two conditions:
\begin{itemize}
	\item $S_t\subseteq S_{t-1}$.
	\item The value $f(y)$ are identical for any $y\in\mathrm{DOWN}(S_t, \tilde{x}, t)$.
\end{itemize}
We then make the following claim, and prove it.

\begin{claim}\label{clm:ssize}
	\[|S_t| = \Omega\left(\log^{\circ (t - 1)}(|S_{t-1}|)\right).\]
\end{claim}

Our proof relies on Erd\H os--Rado Theorem \cite{Erdos1952Combinatorial} from Ramsey theory on hypergraphs. 

\begin{definition}[$k$-Uniform Hypergraph Ramsey Number]
	Suppose $V$ is a set of vertices, and all size-$k$ subsets of $V$ forms $\mathcal{F}_k(V)$. If $E\subseteq\mathcal{F}_k(V)$, then we call $(V,E)$ as a $k$-uniform hypergraph of order $|V|$. Naturally, we call $(V,\mathcal{F}_k(V))$ a complete $k$-uniform hypergraph. 
	
	If the following property holds for complete $k$-uniform hypergraph of order $M$ but not $M-1$, then $r_k(s,t):=M$ is called the $k$-uniform hypergraph Ramsey number: 
	color every $k$-hyperedge red or blue arbitrarily, then there must exist a complete red hyper-subgraph of order $s$, or a complete blue hyper-subgraph of order $t$. 
\end{definition}

\begin{theorem}[\cite{Erdos1952Combinatorial}]\label{thm:erdos}
    $r_2(s,t)\leq\binom{s+t-2}{t-1}$ and
	\[r_k(s, t) \leq 2^{\binom{r_{k-1}(s - 1, t - 1)}{k - 1}},\quad (k>2).\]
\end{theorem}

This theorem implies the following fact: if we color $k$-uniform hypergraph of order $n$ with two colors, then there exists a monochromatic clique of size $\Omega(\log^{\circ k}(n))$.

\begin{proof}[Proof of \clm{ssize}]
	Construct the following complete $t$-uniform hypergraph. The vertex set is $S_{t-1}$. For any $x\in\mathrm{DOWN}(S_{t-1}, \tilde{x}, t)$, we color the hyperedge $\{i:x_i\neq \tilde{x}_i\}$ with red if $f(x)=1$, or blue otherwise. Take the largest monochromatic clique and suppose its vertex set is $S$. According to \thm{erdos}, $|S|=\Omega(\log^{\circ k}(n))$. Furthermore, the monochromaticity implies the value $f(z)$ is identical for any $z\in\mathrm{DOWN}(S, \tilde{x}, t)$. Hence $S$ satisfies the desired conditions where our claim follows immediately. 
\end{proof}

With the help of the claim and notations above, we can complete our proof of \lem{embed}. 

\begin{proof}[Proof of \lem{embed}]
	By invoking \clm{ssize} iteratively, 
	\[|S_t| \geq Z\cdot\log^{\circ ((t-1)t/2 + 1)}(n),\]
	where $Z$ is a positive constant irrelevant to $n$. Take $t'=t'(n):=\lfloor\sqrt{\log^*(n) - 2}\rfloor$. Then $(t'-1)t'/2 + 1 < \log^*(n)$ and $|S_{t'}|=\omega(1)$. In addition, define
	\[r=r(n):= \min\left\{Z\cdot \log^{\circ(t'(t' - 1)/2 + 1)}(n), t'(n)\right\} = \omega(1).\]
	Then we have $r(n)\leq |S_{r(n)}|$. This is because
	\begin{itemize}
		\item if $t'(n)\leq Z\cdot \log^{\circ(t'(t' - 1)/2 + 1)}(n)$, then we have $r(n)=t'(n)\leq Z\cdot \log^{\circ(r(r - 1)/2 + 1)}(n)\leq|S_{r(n)}|$;
		\item if $t'(n)> Z\cdot \log^{\circ(t'(t' - 1)/2 + 1)}(n)$, then $t'(n)>r(n)$, so $r(n)=Z\cdot \log^{\circ(t'(t' - 1)/2 + 1)}(n)\leq Z\cdot \log^{\circ(r(r - 1)/2 + 1)}(n)\leq|S_{r(n)}|$. 
	\end{itemize}
	
	Now take a size $r(n)$ subset $T$ of $S_{r(n)}$ arbitrarily. Define the mapping $\sigma:[n]\backslash T\to\{0,1\}$ such that $\sigma(i)=\tilde{x}_i$. Restrict $f$ over $\sigma$ to obtain a new function $g:=f|_\sigma$. We will prove $g$ is symmetric and non-trivial, and then the lemma follows immediately. 
	
	\subparagraph*{Symmetric.} Assume $x,y\in\{0,1\}^{r(n)}$ such that $|x|=|y|$. Define $x'$ (resp. $y'$) to be the string of size $n$ obtained from $x$ (resp. $y$) and $\sigma$. Recall the definition of $S_{r(n)}$. That is, for any $i\in S_{r(n)}$, it holds that $\tilde{x}_i=1$. Therefore, 
	\[x',y'\in\mathrm{DOWN}(S_{r(n)}, \tilde{x}, |x|)\subseteq\mathrm{DOWN}(S_{|x|}, \tilde{x}, |x|).\]
	By definition of $\mathrm{DOWN}$, we have $f(x')=f(y')$ i.e., $g(x)=g(y)$.
	
	\subparagraph*{Non-trivial.} Assume $z,w\in\{0,1\}^{r(n)}$ such that $|z|=0$ and $|w|=1$. We define $z'$ and $w'$ similarly to $x'$ and $y'$. Suppose $z'=w'\oplus e_i$. As $i\in T\subseteq S_0$, the $i$-th bit is sensitive. Therefore, $f(z')\neq f(w')$ i.e., $g(z)\neq g(w)$. 
\end{proof}

\section[Lower Bounds of degpq(f) for Symmetric Functions]{Lower Bounds of $\deg_{pq}(f)$ for Symmetric Functions}\label{sec:degpq}

We first go further with the analysis of Mahler coefficients of $\mathsf{MOD}$ functions, then prove \thm{degpq_sym} in \sec{degpq_bound}, and give an instance in \sec{degpq_tight}, showing one can never improve the constant factor $1/2$.  
\subsection[More Analyses of MOD and Its Mahler Coefficients]{More Analyses of $\mathsf{MOD}$ and Its Mahler Coefficients}

Let $p$ be a prime and $t$ be a positive integer. Consider over the ring $\mathbb{Z}_{p}$. Recall the notation $\alpha_\ell^{(a,p^t)}$: it is the $\ell$-th Mahler coefficient of $\mathsf{MOD}_n^{a,p^t}$. Since $\deg_p(\mathsf{MOD}_n^{a,p^t})=p^t-1$, the following $p^t\times p^t$ matrix collects all the coefficients of $\mathsf{MOD}_n^{a,p^t}$:
\[{\bm A}_{p^t}:=
\begin{bmatrix}
\alpha_0^{(0,p^t)} & \cdots & \alpha_0^{(p^t-1,p^t)} \\
\vdots & \ddots & \vdots \\
\alpha_{p^t-1}^{(0,p^t)} & \cdots & \alpha_{p^t-1}^{(p^t-1,p^t)}
\end{bmatrix}
.\]

In fact, 
\[\left({\bm A}_{p^t}\right)_{i,j}=\alpha_i^{(j,p^t)}=\binom{p^t-1-j}{p^t-1-i}.\]
This is because
\begin{equation}
\sum_{i = 0}^{{p^t} - 1} \binom{{p^t} - 1 - j}{{p^t} - 1 - i} \binom{|x|}{i} = \binom{{p^t} - 1 - j + |x|}{{p^t} - 1} = \left\{\begin{array}{cl}1 & \text{if $|x| \equiv j \pmod{p^t}$,} \\ 0 & \text{otherwise,}\end{array}\right.
\end{equation}
by Vandermonde's convolution.

The matrix ${\bm A}_{p^t}$ has many elegant properties. For example, the following one shows the relationship between ${\bm A}_{p^t}$ and ${\bm A}_{p}$. We use $\otimes$ to denote matrix tensor product. 
\begin{proposition}\label{prop:tensor}
	On the ring $\mathbb{Z}_p$, 
	\[{\bm A}_{p^t}={\bm A}_{p}^{\otimes t}:=\underbrace{{\bm A}_{p}\otimes {\bm A}_{p}\otimes\cdots\otimes{\bm A}_{p}}_{t}.\]
\end{proposition}

\begin{proof}
	Let $i_\ell$ and $j_\ell$ be the representation of $i$ and $j$ in base $p$. Then by Lucas's Theorem, 
	\[
	\binom{p^t - 1 - j}{p^t - 1 - i} \equiv \dbinom{\sum_{\ell = 0}^{t - 1} (p - 1 - j_\ell) \cdot p^\ell}{\sum_{\ell = 0}^{t - 1} (p - 1 - i_\ell) \cdot p^\ell} \equiv \prod_{\ell = 0}^{t - 1} \binom{p - 1 - j_\ell}{p - 1 - i_\ell} \equiv \prod_{\ell = 0}^{t - 1} (\bm A_p)_{i_\ell,j_\ell} \pmod p.
	\]
\end{proof}

Below we give another observation, which assists with our proof of \thm{degpq_sym}. 

\begin{lemma}\label{lem:apv}
	Suppose $p$ is a prime, and $n<p-1$ is a positive integer. Then for any $v\in\{0,1\}^p$ satisfying $v_i\neq v_j$ for some $0\leq i<j\leq n$, there exists $\lfloor n/2\rfloor+1\leq\ell\leq n$ such that $({\bm A}_p {\bm v})_\ell\neq 0$.
\end{lemma}

Our proof of \lem{apv} utilizes the following proposition on another binomial coefficient matrix.

\begin{proposition}\label{prop:binom2}
	For any prime $p$, integers $j, k$ with $j + k < p$ and distinct $a_0, \ldots, a_k \in \mathbb{F}_p$ satisfying $a_0, \ldots, a_k \geq j$, the matrix
	\[
	\begin{bmatrix}
	\binom{a_0}{j} & \binom{a_1}{j} &\cdots & \binom{a_k}{j} \\
	\binom{a_0}{j+1} & \binom{a_1}{j+1} &\cdots & \binom{a_k}{j+1} \\
	\vdots & \vdots & \ddots & \vdots \\
	\binom{a_0}{j+k} & \binom{a_1}{j+k} &\cdots & \binom{a_k}{j+k} \\
	\end{bmatrix}
	\]
	is invertible over $\mathbb F_p$.
\end{proposition}

\begin{proof}
	One can verify that
	\begin{align*}
	\!\!\!\!\!\text{diag}\left\{\frac{(j + 0)^{\underline 0}}{\binom{a_0}{j}}, \ldots, \frac{(j + k)^{\underline k}}{\binom{a_k}{j}}\right\}\cdot S \cdot \begin{bmatrix}
	\binom{a_0}{j} &\cdots & \binom{a_k}{j} \\
	\vdots & \ddots & \vdots \\
	\binom{a_0}{j+k} & \cdots & \binom{a_k}{j+k}
	\end{bmatrix}
	= \begin{bmatrix}
	(a_0 - j)^0 & \cdots & (a_k - j)^0\\
	\vdots &  \ddots & \vdots \\
	(a_0 - j)^k & \cdots & (a_k - j)^k
	\end{bmatrix},
	\end{align*}
	where $S$ is the second Stirling number matrix, i.e., $S_{ij}=\stirling{i}{j}$, and the notation $x^{\underline y}$ stands for the falling factorial power $x(x-1)\cdots(x-y+1)$. The Vandermonde matrix on the R.H.S. is also invertible because $a_0, \ldots, a_k$ are distinct.
\end{proof}

\begin{proof}[Proof of \lem{apv}]
	Assume towards contradiction that there exists some $v$ satisfying the condition, but $({\bm A}_p v)_\ell=0$ for all $\lfloor n/2\rfloor+1\leq\ell\leq n$. In other words, if we take row $\lfloor n/2\rfloor+1$ to $n$ and column $0$ to $n$ from ${\bm A}_p$ to get another $\lceil n/2\rceil\times (n+1)$ matrix ${\bm B}$ i.e.,
	\[
	\bm B := \begin{bmatrix}
	\binom{p-1-0}{p-1-(\lfloor n/2\rfloor + 1)} & \binom{p-1-1}{p-1-(\lfloor n/2\rfloor + 1)} & \ldots & \binom{p-1-n}{p-1-(\lfloor n/2\rfloor + 1)}\\
	\vdots & \vdots & \ddots & \vdots\\
	\binom{p-1-0}{p-1-n} & \binom{p-1-1}{p-1-n} & \ldots & \binom{p-1-n}{p-1-n}
	\end{bmatrix},
	\]
	then ${\bm B}{\bm v}'={\bm 0}$ where ${\bm v}'=\{v_0,...,v_n\}^T$. (This is because $\left({\bm A}_p\right)_{i,j}=\binom{p-1-j}{p-1-i}=0$ for all $\lfloor n/2\rfloor+1\leq i\leq n$ and $n+1\leq j\leq p-1$. ) 
	
	Next, for any $t\in\left[(\lfloor n/2\rfloor + 1),n\right]$, the sum of row $t$ is $\binom{p-1-0}{p-1-t}+\binom{p-1-1}{p-1-t}+\cdots+\binom{p-1-n}{p-1-t}=\binom{p}{p-t}\equiv 0\pmod{p}$. Therefore ${\bm B}{\bm 1}={\bm 0}$, so we can assume the number of $1$'s in ${\bm v}'$ is no more that $\lceil n/2\rceil$, without loss of generality. (Otherwise, subtract ${\bm v}'$ from ${\bm 1}$.) This means that we can take $s\leq \lceil n/2\rceil$ column vectors of ${\bm B}$, the summation of which is ${\bm 0}$, and furthermore, the last $s$ dimensions of these vectors form a singular matrix with form ${\bm B}'_{i,j}=\binom{a_j}{p-1-n+s-1+i}$. However, by flipping it upside down and applying \prop{binom2}, this matrix is invertible.
\end{proof}

\subsection[Proof of Theorem 3]{Proof of \thm{degpq_sym}}\label{sec:degpq_bound}

Our proof requires the following two lemmas. The first one is often referred to as Periodicity Lemma. It says any function with coprime periods is constant, if the domain is large enough. 

\begin{lemma}[Periodicity Lemma, \cite{fine1965uniqueness}]
	\label{lem:period}
	Let $g$ be an $a$-periodic and $b$-periodic function on domain $\{0,1,\ldots,n\}$ with $gcd(a,b)=1$ and $n\ge a+b-2$. Then $g$ is a constant function.
\end{lemma}

The next one can be regarded as a stronger version of \thm{degpk_sym} with $k=1$. 

\begin{lemma}
	\label{lem:degpq_main}
	Assume $p$ is a prime. For any non-trivial symmetric $f: \{0, 1\}^n \to \{0, 1\}$, 
	\[\deg_p(f) \geq \min\left\{\frac{n}{2}, \left(1-\frac{1}{p}\right)\pi_p(f)\right\}.\]
\end{lemma}

Note that the base-$p$ period appear explicitly in the lower bound. This allows us to apply \lem{period}. We prove \thm{degpq_sym} first. 

\begin{proof}[Proof of \thm{degpq_sym}]
    If $\max\{\deg_p(f), \deg_{q}(f)\}\ge\frac{n}{2}$, the theorem then follows naturally. Otherwise, according to \fac{divide_dominate} and \lem{degpq_main}, we have
    \begin{equation}\label{equ:degpq}
    \deg_{m}(f)\geq\deg_{pq}(f)=\max\{\deg_p(f),\deg_q(f)\}\ge\max\left\{\left(1-\frac{1}{p}\right)\pi_p(f),\left(1-\frac{1}{q}\right)\pi_q(f)\right\}.
    \end{equation}
    On the other hand, the non-triviality of $f(x)$ implies $\pi_p(f)+\pi_q(f)>n+2$ owing to \lem{period}. The last term in Inequality (\ref{equ:degpq}) is $>\frac{n+2}{2+1/(p-1)+1/(q-1)}>\frac{n}{2+1/(p-1)+1/(q-1)}$, and hence the theorem is also true. 
\end{proof}

It remains to show why \lem{degpq_main} is true.

\begin{proof}[Proof of \lem{degpq_main}]
    Consider over the ring $\mathbb{Z}_{p}$. Suppose $\pi_p(f)=p^t$. We write $f(x)$ as we did in Equation (\ref{equ:fexpand}), and let $\alpha_\ell$ be the $\ell$-th Mahler coefficient of $f(x)$. Then $\sum_{j=0}^{p^t-1}w_j\alpha_{\ell}^{(j,p^t)}=\alpha_\ell$, or
	\begin{equation}\label{equ:avsw}
	{\bm \alpha}={\bm A}_{p^t}{\bm w}
	\end{equation}
	if we set ${\bm w}:=\left(w_0,...,w_{p^t-1}\right)^{\top}$ and ${\bm \alpha}:=\left(\alpha_0,...,\alpha_{p^t-1}\right)^{\top}$. 
	
	Divide ${\bm w}$ and ${\bm \alpha}$ into blocks of length $p^{t-1}$ as ${\bm w} = ({\bm w}^{\langle0\rangle}, \ldots, {\bm w}^{\langle p-1\rangle})^\top$ and ${\bm \alpha} = ({\bm \alpha}^{\langle0\rangle}, \ldots, {\bm \alpha}^{\langle0\rangle})^\top$ where ${\bm w}^{\langle i\rangle} \in \{0, 1\}^{p^{t - 1}}, {\bm \alpha}^{\langle i\rangle} \in \mathbb{F}_p^{p^{t - 1}}$. By \prop{tensor}, we have
	\begin{equation}\label{equ:separate}
	{\bm \alpha}^{\langle i\rangle} = \bm A_{p^{t - 1}}\sum_{j = 0}^{p - 1} \left((\bm A_{p})_{ij}{\bm w}^{\langle j\rangle} \right).
	\end{equation}
	
	Consider two cases. One deals with the case $\pi_p(f)=p^t<n$, where we show $\deg_p(f)>\frac{p - 1}{p}\cdot\pi_p(f)$; another deals with $\pi_p(f)\geq n$, where we can obtain $\deg_p(f)\geq n/2$. 
	
	\subparagraph*{Case I ($p^t<n$).} First, assume ${\bm \alpha}^{\langle p-1\rangle} = \bm 0$. Note that ${\bm A}_{p^{t-1}}$ is full-rank according to \prop{binom2}, which allows Equation (\ref{equ:separate}) to be transformed into
	\begin{equation*}
	\bm A_{p^{t - 1}}^{-1}{\bm \alpha}^{\langle p-1\rangle} = \sum_{j = 0}^{p - 1} \left((\bm A_{p})_{p-1,j}{\bm w}^{\langle j\rangle} \right).
	\end{equation*}
	Since $(\bm A_{p})_{p-1,j}=1$, we have $\sum_{j = 0}^{p - 1} {\bm w}^{\langle j\rangle}={\bm 0}$. This implies ${{\bm w}^{\langle 0\rangle}=\cdots={\bm w}^{\langle p-1\rangle}}$ as ${\bm w}^{\langle i\rangle} \in \{0, 1\}^{p^{t - 1}}$, and further, $f(x)$ becomes $p^{t-1}$-periodic, conflicting with $\pi_p(f)=p^t$. Eventually, we have ${\bm \alpha}^{\langle p-1\rangle} \neq \bm 0$. Because $n> p^t$, the highest non-zero Mahler coefficient indicates the degree of $f$ (see the remark below \fct{degm_mahler}), and then $\deg_p(f)>(p-1)p^{t-1}=\frac{p - 1}{p}\cdot\pi_p(f)$.
	
	\subparagraph*{Case II ($p^t\geq n$).} By Equation (\ref{equ:avsw}), 
	\begin{equation}\label{equ:step1}
	\left({\bm I}_p\otimes ({\bm A}_{p^{t - 1}})^{-1}\right){\bm A}_{p^t}{\bm w}=\left({\bm I}_p\otimes ({\bm A}_{p^{t - 1}})^{-1}\right) {\bm \alpha}.
	\end{equation}
	The R.H.S. of (\ref{equ:step1}) is just
	\[\left({\bm I}_p\otimes ({\bm A}_{p^{t - 1}})^{-1}\right) {\bm \alpha}=\begin{bmatrix}
	(\bm A_{p^{t - 1}})^{-1}{\bm \alpha}^{\langle 0\rangle}\\
	\vdots\\
	(\bm A_{p^{t - 1}})^{-1}{\bm \alpha}^{\langle p-1\rangle}
	\end{bmatrix}=:
	\begin{bmatrix}
	{\bm \beta}^{\langle 0\rangle}\\
	\vdots\\
	{\bm \beta}^{\langle p-1\rangle}
	\end{bmatrix}
	.\]
	The L.H.S. of (\ref{equ:step1}) can be written as
	\begin{align*}
	\left({\bm I}_p\otimes ({\bm A}_{p^{t - 1}})^{-1}\right){\bm A}_{p^t}{\bm w}&=\left({\bm I}_p\otimes ({\bm A}_{p^{t - 1}})^{-1}\right)({\bm A}_{p}\otimes {\bm A}_{p^{t-1}}){\bm w}\\
	&=({\bm I}_p{\bm A}_{p})\otimes(({\bm A}_{p^{t - 1}})^{-1}{\bm A}_{p^{t-1}}){\bm w}\\
	&=({\bm A}_{p}\otimes {\bm I}_{p^{t-1}}){\bm w}.
	\end{align*}
	Therefore,
	\begin{equation}\label{equ:step2}
	({\bm A}_{p}\otimes {\bm I}_{p^{t-1}}){\bm w}=
	\left({\bm \beta}^{\langle 0\rangle},\cdots,{\bm \beta}^{\langle p-1\rangle}\right)^{\top}.
	\end{equation}
	For $0\leq j< p^{t-1}$ we define
	\[
	\tilde {\bm \beta}^{\langle j\rangle} := 
	\left({\bm \beta}^{\langle 0\rangle}_j,\cdots,{\bm \beta}^{\langle p-1\rangle}_j\right)^{\top}
	\text{ and }
	\tilde {\bm w}^{\langle j\rangle} := 
	\left({\bm w}^{\langle 0\rangle}_j,\cdots,{\bm w}^{\langle p-1\rangle}_j\right)^{\top},
	\]
	Intuitively, vectors with tildes here contain entries taken from the original vector with stride $p^{t-1}$. Then Equation (\ref{equ:step2}) implies
	\[{\bm A}_p\tilde{\bm w}^{\langle j\rangle}=\tilde {\bm \beta}^{\langle j\rangle}.\]
	
	Let $n' = \lfloor(n + 1)/p^{t - 1}\rfloor - 1$, $n'' = \lceil(n + 1)/p^{t - 1}\rceil - 1$, and $m' = n \bmod p^{t - 1}$. Consider the following two subcases:
	
	\subparagraph*{Subcase II-1. } Suppose there exists $\ell \leq m'$ and $i, j \in [0, n'']$ such that $\tilde{\bm w}^{\langle\ell\rangle}_i \neq \tilde{\bm w}^{\langle\ell\rangle}_j$. According to \lem{apv}, there exists $i' \in [\lfloor n''/2 \rfloor + 1, n'']$ satisfying $0\neq\left({\bm A}_p\tilde{\bm w}^{\langle j\rangle}\right)_{i'}=\tilde {\bm \beta}^{\langle\ell\rangle}_{i'}={\bm \beta}^{\langle i'\rangle}_{\ell}$. Because $\bm A_{p^{t-1}}$ is invertible, we have ${\bm \alpha}^{\langle i'\rangle} = \bm A_{p^{t-1}}{\bm \beta}^{\langle i'\rangle}\neq{\bm 0}$. What's more, 
	\begin{itemize}
		\item if $i' < n''$ and recall \fct{degm_mahler}, we have $\deg_p(f) \geq (\lfloor n''/2 \rfloor + 1)\cdot p^{t-1} \geq n/2$;
		\item if $i' = n''$, we select the minimum $\ell$ such that ${\bm \beta}^{\langle i'\rangle}_\ell \neq 0$. Due to the fact $(\bm A_{p^{t - 1}})_{\ell,j}=\binom{p^{t - 1}-1-j}{p^{t - 1}-1-\ell}=0$ when $j>\ell$, it follows that
		\begin{equation}\label{equ:caseii1}
		\bm\alpha^{\langle i'\rangle}_\ell = \sum_{j = 0}^{\ell} (\bm A_{p^{t - 1}})_{\ell,j} \bm\beta^{\langle i'\rangle}_j = (\bm A_{p^{t - 1}})_{\ell,\ell} \bm\beta^{\langle i'\rangle}_\ell \neq 0.
		\end{equation}
		Eventually $\deg_p(f) \geq n''\cdot p^{t - 1} \geq n/2$.
	\end{itemize}

	\subparagraph*{Subcase II-2. } Otherwise, there exists a minimum $\ell \in [m' + 1, p^{t - 1} - 1]$ and $i,j \in [0, n']$ such that $\tilde {\bm w}^{\langle\ell\rangle}_i \neq \tilde {\bm w}^{\langle\ell\rangle}_j$ as $f(x)$ is non-trivial. The same argument shows that there exists $i' \in [\lfloor n'/2 \rfloor + 1, n']$ satisfying ${\bm \beta}^{\langle i'\rangle}_{\ell}\neq 0$. In addition, the condition $\tilde {\bm w}^{\langle\ell'\rangle}_0 = \cdots = \tilde {\bm w}^{\langle\ell'\rangle}_{n'}$ for all $\ell'<\ell$ implies
	\[
	\bm\beta_{\ell'}^{\langle i'\rangle} = \tilde{\bm\beta}^{\langle \ell'\rangle}_{i'} = \sum_{j = 0}^{p - 1}(\bm A_p)_{i',j}\tilde {\bm w}^{\langle\ell'\rangle}_j = {\bm w}^{\langle\ell'\rangle}_0 \cdot \sum_{j = 0}^{i'}\binom{p-1-j}{p-1-i'}={\bm w}^{\langle\ell'\rangle}_0 \cdot\binom{p}{p-i'} = 0.
	\]
	Hence, by imitating (\ref{equ:caseii1}) we can obtain ${\bm \alpha}^{\langle i'\rangle}_\ell \neq 0$, which leads to $\deg_p(f) \geq \lfloor n'/2 \rfloor \cdot p^{t - 1} + m' + 1 \geq n/2$.
\end{proof}

\subsection[Proof of Theorem 4]{Proof of \thm{degpq_2}}\label{sec:degpq_tight}

We will later apply the following lemma about Diophantine approximation, which is an immediate corollary of Kronecker's Theorem.

\begin{lemma} \label{lem:kro}
    Suppose real numbers $a_1, \dots, a_k$ satisfy that $1,a_1, \dots, a_k$ are linearly independent over $\mathbb{Q}$. Then, for any $\varepsilon > 0$, there exists infinitely many positive integers $\ell$ such that $\ell a_i \bmod 1 \in (1-\varepsilon,1)$ for each $i = 1,\dots,k$.
\end{lemma}

Now we prove \thm{degpq_2}.
\begin{proof}[Proof of \thm{degpq_2}]
    Write $m = p_1 p_2 \cdots p_k$ for $p_i$ being primes.
    Choose a prime $q$ different from all $p_i$.
    Fix an arbitrary $\varepsilon > 0$.
    Let $a_i = \log q / \log p_i$ for $i = 1,\dots,k$.
    Then $1,a_1, \dots, a_k$ are linearly independent over $\mathbb Q$, otherwise a nontrivial linear relation can be exponentiated to contradict the unique factorization theorem over $\mathbb{Z}^+$.
    Applying \lem{kro} we get infinitely many $\ell$ that satisfy the condition $\ell\cdot \log q / \log p_i \bmod 1 \in (1-\varepsilon,1)$, which implies $p_i^{r_i} / q^\ell \in (1, p_i^\varepsilon)$ where $r_i = \lceil \ell \log q / \log p_i \rceil$.
    
    Now, choose a sufficiently large $\ell$, let $n = 2q^\ell$ and define $f: \{0,1\}^n \to \{0,1\}$ by $f(x) =  \mathbb{I}[|x|=q^\ell]$.
    Then $f$ is $p_i^{r_i}$-periodic since $p_i^{r_i} > q^\ell$.
    Therefore $\deg_{p_i}(f) \leq p_i^{r_i} - 1$ by \thm{wilson1}.
    Thus, \[\deg_m(f) \leq \max_{1 \leq i \leq k} \{ p_i^{r_i} \} \leq \frac{n}{2} \max_{1 \leq i \leq k} \{ p_i^{\varepsilon} \}.\]
    The theorem follows by letting $\varepsilon \to 0$.
\end{proof}

\section{Conclusion}\label{sec:con}

In a nutshell, we explore and exploit the matrices consisting of Mahler coefficients of the $\mathsf{MOD}$ function, serving as a significant extension of Wilson's arguments. This approach fully characterizes the modulo degree of Boolean functions when the base is prime or prime power, and provides good lower bounds for the composite case with the help of periodicity lemma. In addition, we also show a practical way to generalize properties of symmetric functions to non-degenerate ones by a Ramsey-type argument. 

Nevertheless, there is still ample room for further discussion. First and foremost, we conjecture that the constant factor in \thm{degpq_sym} can be improved to $1/2$ in correspondence with \thm{degpq_2}. Moreover, an anonymous reviewer also raises a good question with regard to \thm{degpk_all}: Could the extraordinary large prerequisite $n \geq \mathrm{tower}(\poly(p,k))$ (which is implicit in the proof) be improved to something like $n \geq \exp(\poly(p,k))$? We also wonder if it is possible to embed other kinds of functions to derive similar results. Above all, both \conj{1} and \conj{2} remain open. 
	
	\section*{Acknowledgements}
	
	We thank Noga Alon and Xiaoxu Guo for pointing us to relevant references~\cite{alon2001lower,fine1965uniqueness}. We also appreciate anonymous reviewers' careful feedback and constructive advice.
	
	\bibliography{reference}
	
\end{document}